\documentclass[a4paper]{article}

\usepackage{amsmath, amssymb}
\usepackage{graphicx, subfigure}
\usepackage{booktabs}
\usepackage{rotating}

\newcommand{\ExpectMeas}[2]{\mathbb{E}^{#1}\left[#2\right]}

\newcommand{\keywords}[1]{{\bf Keywords:} {#1}}
\newcommand{\subclass}[1]{{\bf Mathematics Subject Classification (2000):} {#1}}
\newcommand{\JEL}[1]{{\bf JEL:} {#1}}
\newcommand{\qed}{\hfill$\Box$}

\newtheorem{thm}{theorem}[section]
\newtheorem{theorem}[thm]{Theorem}
\newtheorem{proposition}[thm]{Proposition}
\newtheorem{corollary}[thm]{Corollary}
\newtheorem{lemma}[thm]{Lemma}
\newenvironment{proof}{{\bf Proof }}{}

\newcommand{\MakeTitle}{\maketitle\newcommand{\and}{$\cdot$ }}

\title{A Family of Maximum Entropy Densities Matching Call Option Prices
	\thanks{
		We would like to thank Iain Clark, Andrey Gal, Alex Langnau and Olivier Taghizadeh for helpful comments and suggestions.
	}
}

\author{
	Cassio Neri
	\thanks{Lloyds Banking Group, \texttt{cassio.neri@lloydsbanking.com}.}
	\and
	Lorenz Schneider
	\thanks{Center for Financial Risks Analysis (CEFRA), EMLYON Business School, \texttt{schneider@em-lyon.com}.}
}

\date{\today}

\begin{document}

\MakeTitle

\begin{abstract}
We investigate the position of the Buchen-Kelly density \cite{BuchenKelly1996} in the family of entropy maximising densities from \cite{NeriSchneider2009} which all match European call option prices for a given maturity observed in the market.
Using the Legendre transform which links the entropy function and the cumulant generating function, we show that it is both the unique continuous density in this family and the one with the greatest entropy.
We present a fast root-finding algorithm that can be used to calculate the Buchen-Kelly density,
and give upper boundaries for three different discrepancies that can be used as convergence criteria.
Given the call prices, arbitrage-free digital prices at the same strikes can only move within upper and lower boundaries given by left and right call spreads.
As the number of call prices increases, these bounds become tighter, and we give two examples where the densities converge to the Buchen-Kelly density in the sense of relative entropy when we use centered call spreads as proxies for digital prices.
As pointed out by Breeden and Litzenberger \cite{BreedenLitzenberger1978}, in the limit a continuous set of call prices completely determines the density.

\bigskip


\bigskip

\subclass{91B24 \and 91B28 \and 91B70 \and 94A17}

\bigskip

\JEL{C16 \and C63 \and G13}
\end{abstract}

\section{Introduction}
\label{Introduction}

An important problem in derivatives valuation consists in finding suitable probability densities for the underlying asset such that observed market quotes are matched under risk-neutral pricing.
In practice, many schemes rely on somehow choosing an interpolation method and hoping that the difference between one choice and another is not too big.
However, since typically the first and second derivatives of the interpolating function have to be taken, these choices inevitably do end up having a big impact on derivative prices.

The concept of entropy provides a clear criterion of how to deal with this problem in the most unbiased way possible.
Making hard to justify ad-hoc assumptions becomes unnecessary, and it is therefore no wonder that this approach is becoming more and more popular in the financial literature (see \cite{AvellanedaFriedmanHolmesSamperi1997}, \cite{BorweinChoksiMarechal2003}, \cite{BrodyBuckleyConstantinou2007}, \cite{BrodyBuckleyConstantinouMeister2005}, \cite{BrodyBuckleyMeister2004}, \cite{BuchenKelly1996}, \cite{DempsterMedovaYang2007}, \cite{Frittelli2000}, \cite{Gulko1999}, \cite{Gulko2002}).

In this paper, we continue our investigation on Maximum Entropy Densities (MEDs) begun in \cite{NeriSchneider2009}.
We showed that, under risk-neutral pricing and non arbitrage conditions, there is a unique density compatible with market prices of European call and digital options which maximises entropy.
This density turns out to be the unique piece-wise exponential density which matches the given market prices.

The first question we had in mind was how the Buchen-Kelly density introduced in \cite{BuchenKelly1996} fit into this framework.
The Buchen-Kelly density is obtained by imposing only the European call prices as constraints and finding the corresponding entropy maximiser. The result is a continuous piece-wise exponential density.

Fixing the call prices and letting the digital prices vary (within certain bounds given by left and right call spreads) leads to a whole family $\mathcal{G}$ of MEDs which all match the call prices.
The probability mass assigned to each interval between two strikes can vary significantly across MEDs which all meet the same constraints imposed by the call prices.

Our first result is to express the entropy $H(g)$ for $g\in\mathcal{G}$ directly in terms of market data (Theorem \ref{H(D)}) given by call and digital prices.
Since the call prices are fixed, we see $H$ as a function of digital prices only.
Using this result, we derive a simple formula that expresses the sensitivity of $H$ with respect to changes in digital prices (Theorem \ref{H'(D)}).
We also relate this sensitivity to the continuity of densities (Corollary \ref{Corollary1}).

Then, by explicitly calculating the Hessian matrix of $H$, we show that $H$ has a unique critical point,
and that the corresponding Buchen-Kelly density is the unique continuous density in $\mathcal{G}$ (Corollary \ref{BKCharacterization}).

Since we have shown that $H$ is concave with respect to digital prices and calculated the gradient and Hessian of $H$ explicitly,
it is straightforward to implement a Newton-Raphson algorithm that finds the zero of the gradient.
In particular, the Hessian matrix is tridiagonal and very easy to invert.
By showing that $H$ is indeed {\it strongly} concave (Proposition \ref{StrongConcavity}) and explicitly calculating an upper boundary for the Hessian's biggest eigenvalue,
we are able to give upper boundaries for the difference between the entropy of the Buchen-Kelly density and that of the density given by
a set of digital prices, and for the Euclidean norm of the difference between the Buchen-Kelly digital prices and the given ones (Proposition \ref{Estimates}).
Also, by using the Csisz\'ar-Kullback inequality together with these results, we find an upper bound for the $L^1$-distance of a given density in $\mathcal{G}$
from the Buchen-Kelly density (Theorem \ref{Th:L1Difference}).
We conclude the section by illustrating the algorithm with two low-dimensional examples, and by pointing out some of its advantages over the
original algorithm presented in \cite{BuchenKelly1996}.

Finally, we study how the family $\mathcal{G}$ becomes more constrained as the set of strikes increases.
Given the call prices, arbitrage-free digital prices at the same strikes can only move within upper and lower boundaries given by left and right call spread prices.
As the set of strikes increases, and call prices get closer to one another, these bounds become tighter, and we give two examples, one fictitious and one with call option data on the S\&P, where the densities converge to the Buchen-Kelly density in the sense of relative entropy when we use centered call spreads as proxies for digital prices. In the limit, as pointed out by Breeden and Litzenberger \cite{BreedenLitzenberger1978}, a continuous set of call prices completely determines the density.

\section{The MED Obtained from Call Prices - The Buchen-Kelly Density}
\label{sec:BuchenKellyDensity}

For a fixed maturity $T$, assume prices $C(K_i)$ of European call options on a particular asset are observed in the market for a set of $n$ strikes $0 < K_1 < ... < K_n$.
Denote by $\tilde{C}_i := C(K_i) / DF(T)$ the undiscounted prices of such calls, where $DF(T)$ is the discount factor from today to maturity $T$.

For notational convenience, we introduce the ``strikes'' $K_0 := 0$ and $K_{n+1} := \infty$.
We set $\tilde{C}_0$ to the $T$-forward price of the asset and $\tilde{C}_{n+1} := 0$.

Let $g$ be a strictly positive density over $[0,\infty[$ for $S(T)$ -- the underlying asset price at time $T$ -- that under risk-neutral pricing matches the call prices, that is,
\begin{equation}
\ExpectMeas{g}{(S(T) - K_i)^+} = \int_{K_i}^\infty (x - K_i) g(x) dx = \tilde{C}_i, \quad \forall i=0,...,n.
\label{CallConstraints}
\end{equation}

It is well known that call prices implied by $g$ decrease with strikes.
More precisely, $\ExpectMeas{g}{(S(T) - K)^+}$ is strictly decreasing on $K\in [0,\infty[$.
Therefore, the following non-arbitrage condition holds:
\begin{equation}
\tilde{C}_i > \tilde{C}_{i+1}, \quad \forall i = 0, ..., n.
\label{NoArbitrageCalls}
\end{equation}

Buchen and Kelly \cite{BuchenKelly1996} have shown that, under \eqref{NoArbitrageCalls}, there exists a unique density $g$ on $[0,\infty[$ that maximises entropy and matches call prices \eqref{CallConstraints}.
In the sequel we refer to this MED as the Buchen-Kelly density.

Using a slightly different notation, Buchen and Kelly show that their MED is the unique density of the form
$g(x)=\mu^{-1} \exp \left( \sum_{i=0}^n \lambda_i (x-K_i)^+ \right)$
which matches the call prices \eqref{CallConstraints}.
They then proceed to find the parameters $\lambda_0$, ..., $\lambda_n$
($\mu$ is a normalization constant which is easily expressed as a function of $\lambda_0$, ..., $\lambda_n$)
through a well posed, but sometimes ill conditioned, multi-dimensional root-finding problem.

We propose a different approach where rather than working on the space of parameters $\lambda_0$, ..., $\lambda_n$
we work on the space of digital prices implied by the densities.
As we shall see, the root-finding problems we encounter are easier and more stable than the original one that Buchen and Kelly studied.

Notice that the Buchen-Kelly density is continuous and piecewise-exponential. More precisely, on each interval $[K_i, K_{i+1}[$, $g$ has the form $g(x) = \alpha_ie^{\beta_i x}$ for some constants $\alpha_i > 0$ and $\beta_i \in \mathbb{R}$.
(See Section \ref{sec:BKNotation}.)

\subsection{Implied Digital Prices}

Under risk neutral pricing, the undiscounted price of a digital option with strike $K>0$ implied by a strictly positive density $g$ over $[0, \infty[$ is given by
\begin{equation}
\label{IntegralForDigital}
\ExpectMeas{g}{ \mathbf{I}_{\{S(T) \ge K\}} } = \int_K^\infty g(x) dx.
\end{equation}

It is also well known that $\ExpectMeas{g}{(S(T) - K)^+}$ is strictly convex on $K\in[0,\infty[$. Furthermore, from Lebesgue's dominated convergence, implied call and digital prices are related by
\begin{equation*}
\ExpectMeas{g}{ \mathbf{I}_{\{S(T) \ge K\}} } = -\frac{d}{dK} \ExpectMeas{g}{(S(T) - K)^+}.
\end{equation*}

Therefore, provided that $g$ matches market call prices \eqref{CallConstraints}, the implied digital prices are bounded by call spread prices:
\begin{equation}
-\frac{\tilde{C}_{i} - \tilde{C}_{i-1}}{K_i - K_{i-1}} >
\ExpectMeas{g}{ \mathbf{I}_{\{S(T) \ge K_i\}} } >
-\frac{\tilde{C}_{i+1} - \tilde{C}_i}{K_{i+1} - K_i}, \qquad \forall i = 1, ..., n.
\label{DigitalBounds}
\end{equation}
(Here, the rightmost quantity for $i=n$ must be read as zero.)

In particular, the digital prices implied by the Buchen-Kelly density are also bounded by these call spread prices.
Since the Buchen-Kelly density has piecewise-exponential form, the integral in \eqref{IntegralForDigital} is straightforward to calculate analytically
\cite{NeriSchneider2009}.

\section{The MED Obtained from Call and Digital Prices}

If in addition to call prices, for the same maturity and set of strikes, the prices $D(K_i)$ of European digitals are observed in the market (e.g. on the S\&P 500 index from the CBOE, where they are called {\it binary} options \cite{CBOE2010}), then we set
$\tilde{D}_i := D(K_i) / DF(T)$, for $i=1, ...,n$, $\tilde{D}_0 := 1$ and $\tilde{D}_{n+1} := 0$.

In this case, from \eqref{DigitalBounds}, we must have
\begin{equation}
-\frac{\tilde{C}_{i} - \tilde{C}_{i-1}}{K_i - K_{i-1}} > \tilde{D}_i > -\frac{\tilde{C}_{i+1} - \tilde{C}_i}{K_{i+1} - K_i}, \qquad \forall i = 1, ..., n,
\label{NoArbitrageDigitals}
\end{equation}

In \cite{NeriSchneider2009} we have shown that under \eqref{NoArbitrageDigitals} there is a unique density $g$ on $[0,\infty[$ which maximises entropy and matches both call prices \eqref{CallConstraints} and digital prices:
\begin{equation}
\ExpectMeas{g}{ \mathbf{I}_{\{S(T) \ge K_i\}} } = \int_{K_i}^\infty g(x) dx = \tilde{D}_i, \quad \forall i = 0, ..., n.
\label{DigitalConstraints}
\end{equation}

In general, $g$ is not continuous but, like the Buchen-Kelly density, it is still piecewise-exponential.

Keeping the call prices fixed, we now let the digital prices $\tilde{D}_1$, ..., $\tilde{D}_n$ vary inside the boundaries given in \eqref{NoArbitrageDigitals}.
We introduce the set $\Omega\subset\mathbb{R}^n$ of all $\tilde{D}=(\tilde{D}_1, ..., \tilde{D}_n)\in\mathbb{R}^n$ verifying \eqref{NoArbitrageDigitals}.
Note that $\Omega$ is an open $n$-dimensional rectangle.

For each $\tilde{D}\in\Omega$ we denote by $g_{\tilde{D}}$ the MED obtained in \cite{NeriSchneider2009}.
(See also Section \ref{sec:Recall} here for a review of some of the main results of \cite{NeriSchneider2009}.)
Let $\mathcal{G}=\{g_{\tilde{D}} \; | \; \tilde{D}\in\Omega\}$ be the family of all MEDs obtained in this way.

The Buchen-Kelly density has the greatest entropy among all densities over $[0,\infty[$ matching call prices \eqref{CallConstraints}, and its implied digital prices verify \eqref{DigitalBounds}. In other words, the Buchen-Kelly is the element of $\mathcal{G}$ with the greatest entropy.

In a typical market, where call prices are not available at a large number of strikes, the upper and lower boundaries for digital prices given in \eqref{DigitalBounds} can be quite far apart, that is, both $\Omega$ and $\mathcal{G}$ can be quite wide.

However, quoting \cite{BuchenKelly1996}, ``as is often the case in practice, we may have many options trading close-to-the-money.
Hence, there could be several closely spaced options, with strike differences small compared to the current asset price''.
In this case, Buchen and Kelly claim that their multi-dimensional root-finding problem may be poorly conditioned and become unstable.

\subsection{How the MED is Obtained from Call and Digital Prices}
\label{sec:Recall}

We shall recall briefly how each $g_{\tilde{D}}\in\mathcal{G}$ is obtained from call and digital prices.
Firstly, we introduce functions $c_0, ..., c_n$ which will appear repeatedly in the following:
\begin{equation}
c_i(\beta) := \left\{
\begin{array}{cl}
\displaystyle \ln \left( \frac{e^{\beta K_{i+1}} - e^{\beta K_i}}{\beta} \right) \quad & \text{for } i < n \text{ and } \beta \neq 0, \\
\\
\displaystyle \ln(K_{i+1} - K_i) & \text{for } i < n \text{ and } \beta = 0,\\
\\
\displaystyle \ln\left(-\frac{e^{\beta K_i}}{\beta}\right) & \text{for } i = n \text{ and } \beta < 0.
\end{array}
\right.
\end{equation}
Their first and second derivatives are given by
\begin{equation*}
c_i'(\beta) = \left\{
\begin{array}{cl}
\displaystyle \frac{K_{i+1} e^{\beta K_{i+1}} - K_i e^{\beta K_i}}{e^{\beta K_{i+1}} - e^{\beta K_i}} - \frac{1}{\beta} \quad & \text{for } i < n \text{ and } \beta \neq 0, \\
\\
\displaystyle \frac{K_{i+1} + K_i}{2} & \text{for } i < n \text{ and } \beta = 0, \\
\\
\displaystyle K_i - \frac{1}{\beta} \quad & \text{for } i = n \text{ and } \beta < 0.
\end{array}
\right.
\end{equation*}
and
\begin{equation*}
c_i''(\beta) = \left\{
\begin{array}{cl}
- (K_{i+1} - K_i)^2
  \dfrac{e^{\beta(K_{i+1} + K_i)}}{(e^{\beta K_{i+1}} - e^{\beta K_i})^2}
+ \dfrac1{\beta^2}
\quad & \text{for } i < n \text{ and } \beta \neq 0, \\
\\
\dfrac{(K_{i+1} - K_i)^2}{12} & \text{for } i < n \text{ and } \beta = 0, \\
\\
\dfrac1{\beta^2} \quad & \text{for } i = n \text{ and } \beta < 0.
\end{array}
\right.
\end{equation*}

Next, we recall the definition and some properties of the Legendre transform of a convex function.
Let $f$ be a convex function on $\mathbb{R}$. Define the {\it Legendre transform} of $f$ as
\begin{equation}
\label{DefinitionLegendreTransform}
f^*(y) := \sup_{x \in \mathbb{R}} \{ xy - f(x) \}, \quad y \in \mathbb{R}.
\end{equation}
We will frequently use the following properties of differentiable convex functions, which we do not prove here.
(See \cite{Arnold1989}, \cite{Ellis1985}, \cite{Rockafellar1970} for more details about the Legendre transform.)
\begin{proposition}
\label{PropertiesLegendreTransform}
Let $f$ be a differentiable and strictly convex function on $\mathbb{R}$.
Then the following conclusions hold.
\begin{enumerate}
\item\label{Young}
$xy \le f(x) + f^*(y)$ \text{with equality holding if, and only if,} $y = f'(x)$.
\item\label{Derivatives}
$[f^*]' = [f']^{-1}$.
\item\label{Twice}
$f^*$ is convex and $[f^*]^* = f$.
\end{enumerate}
\end{proposition}

\begin{figure}[ht]
\centering
\subfigure[Graph of $c_i$.] {
\includegraphics[width=.45\textwidth]{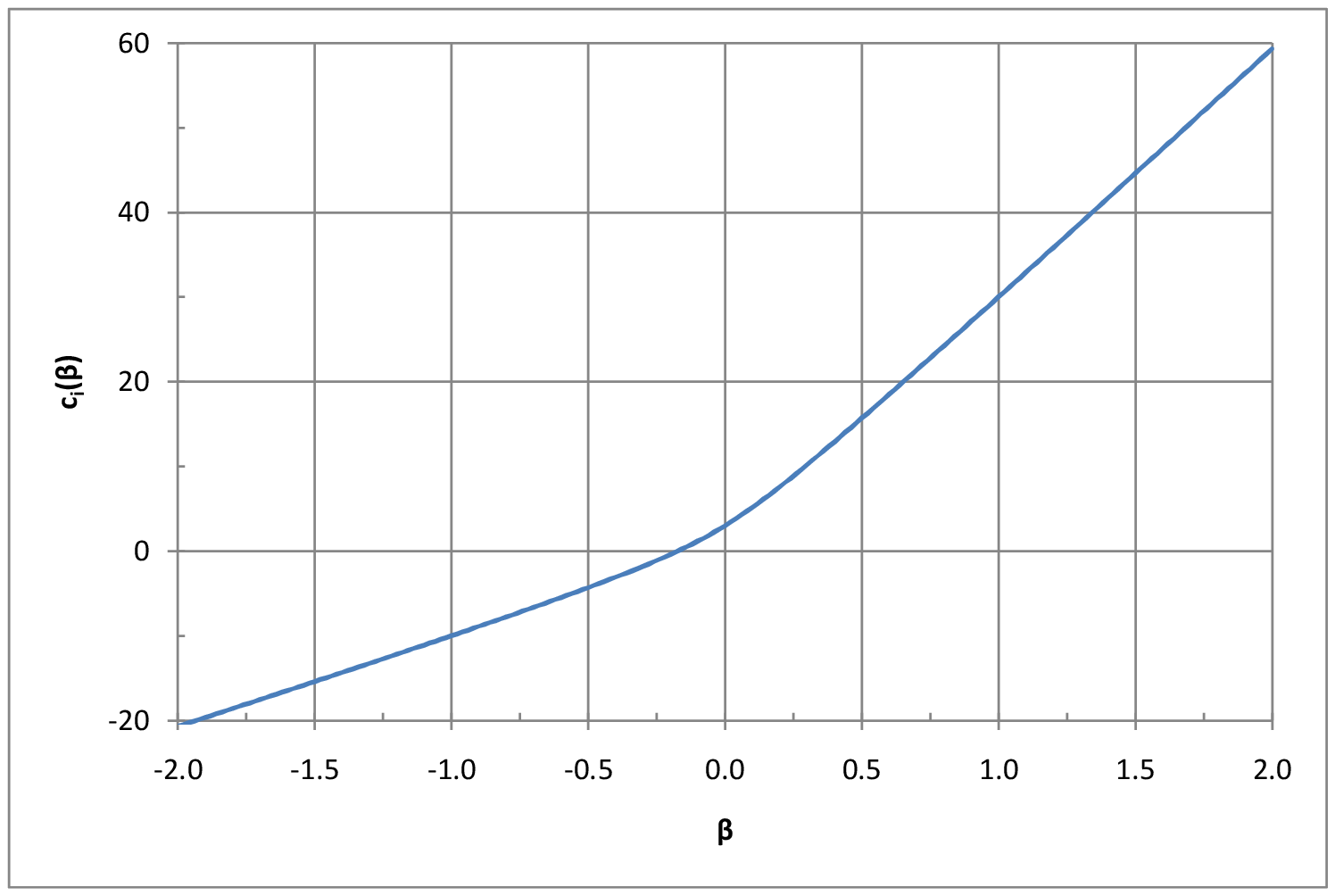}
\label{Fig:1a}
}
\subfigure[Graph of $c_i'$.] {
\includegraphics[width=.45\textwidth]{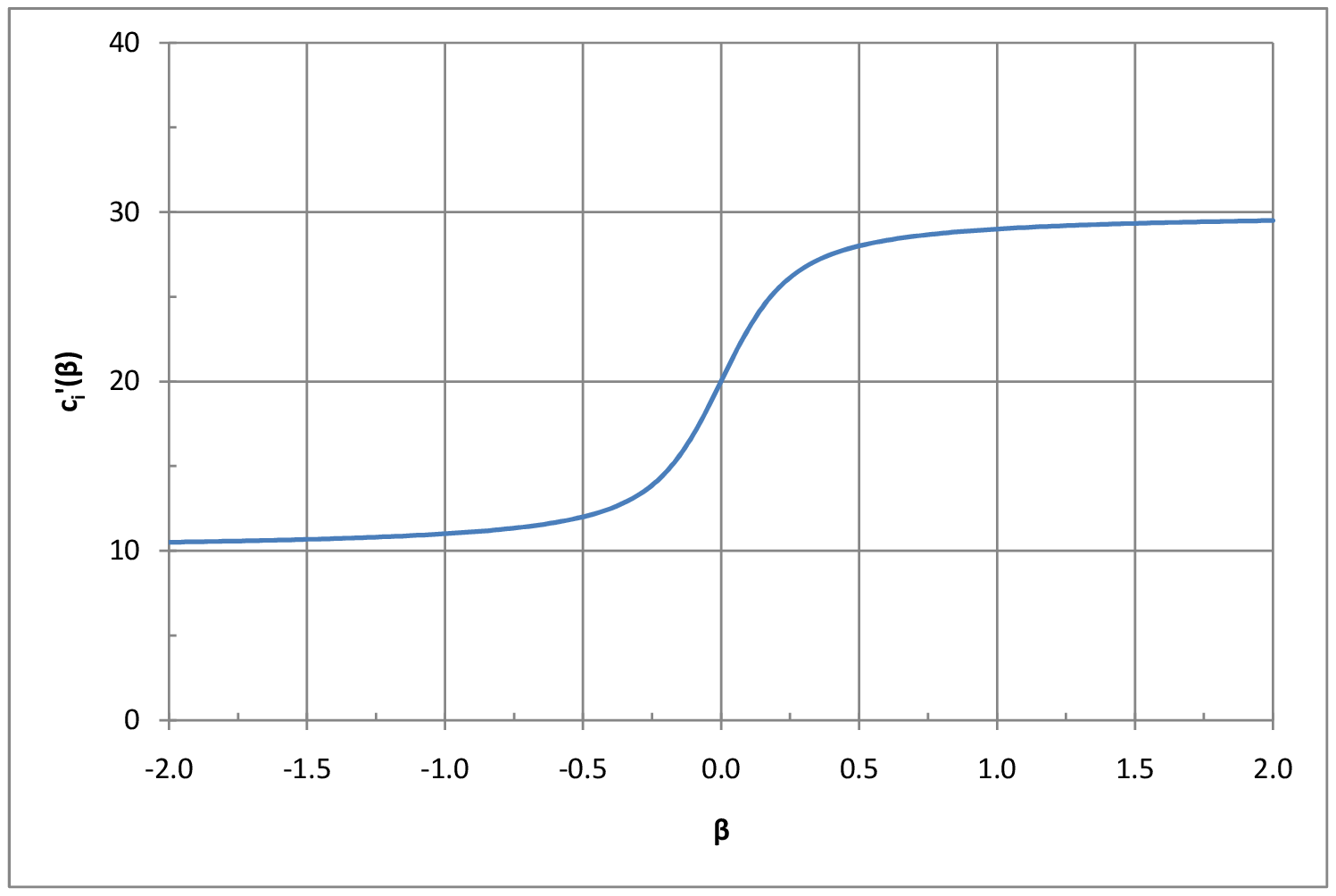}
\label{Fig:1b}
}
\subfigure[Graph of ${[}c_i^*{]}'$.] {
\includegraphics[width=.45\textwidth]{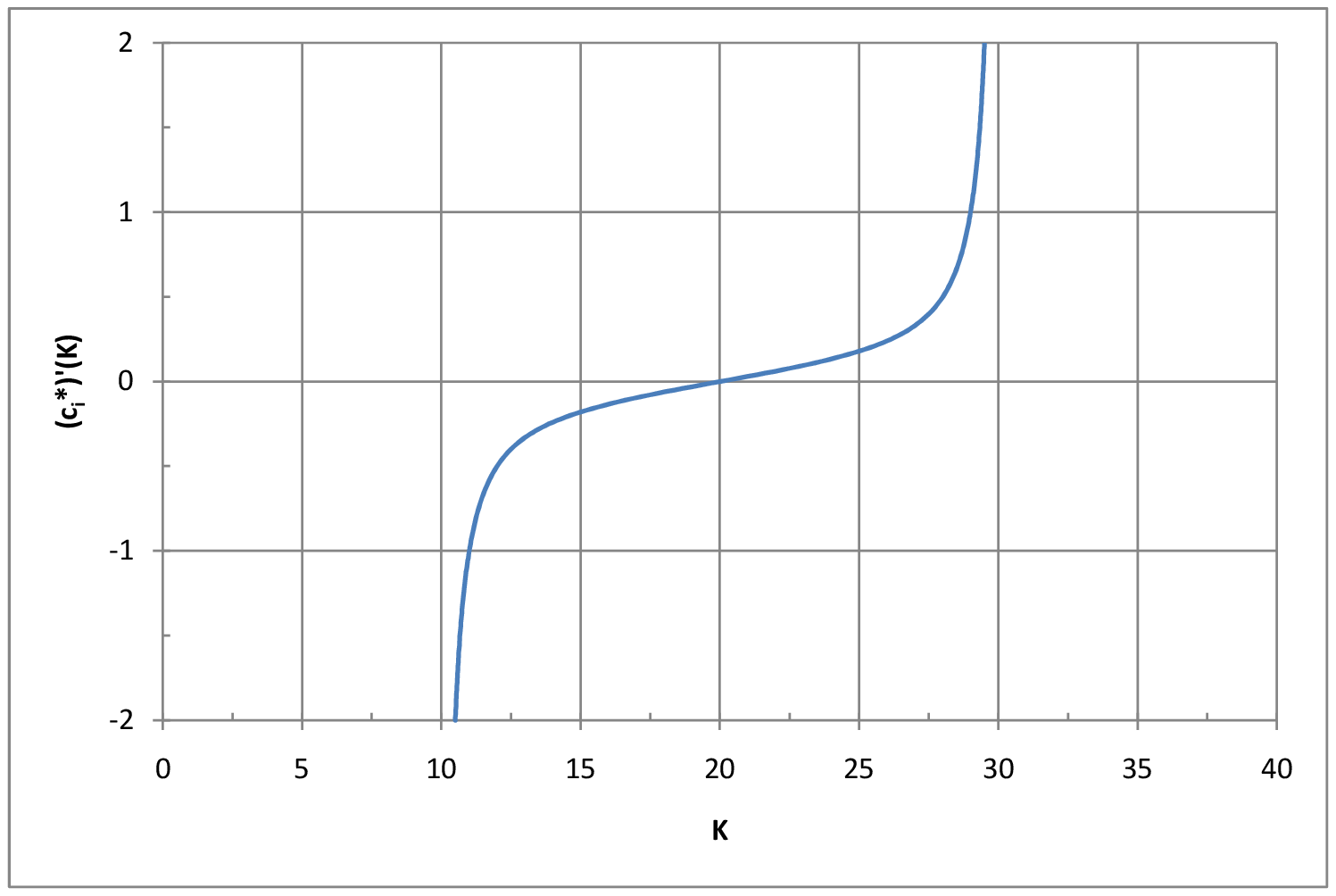}
\label{Fig:1c}
}
\subfigure[Graph of $c_i^*$.] {
\includegraphics[width=.45\textwidth]{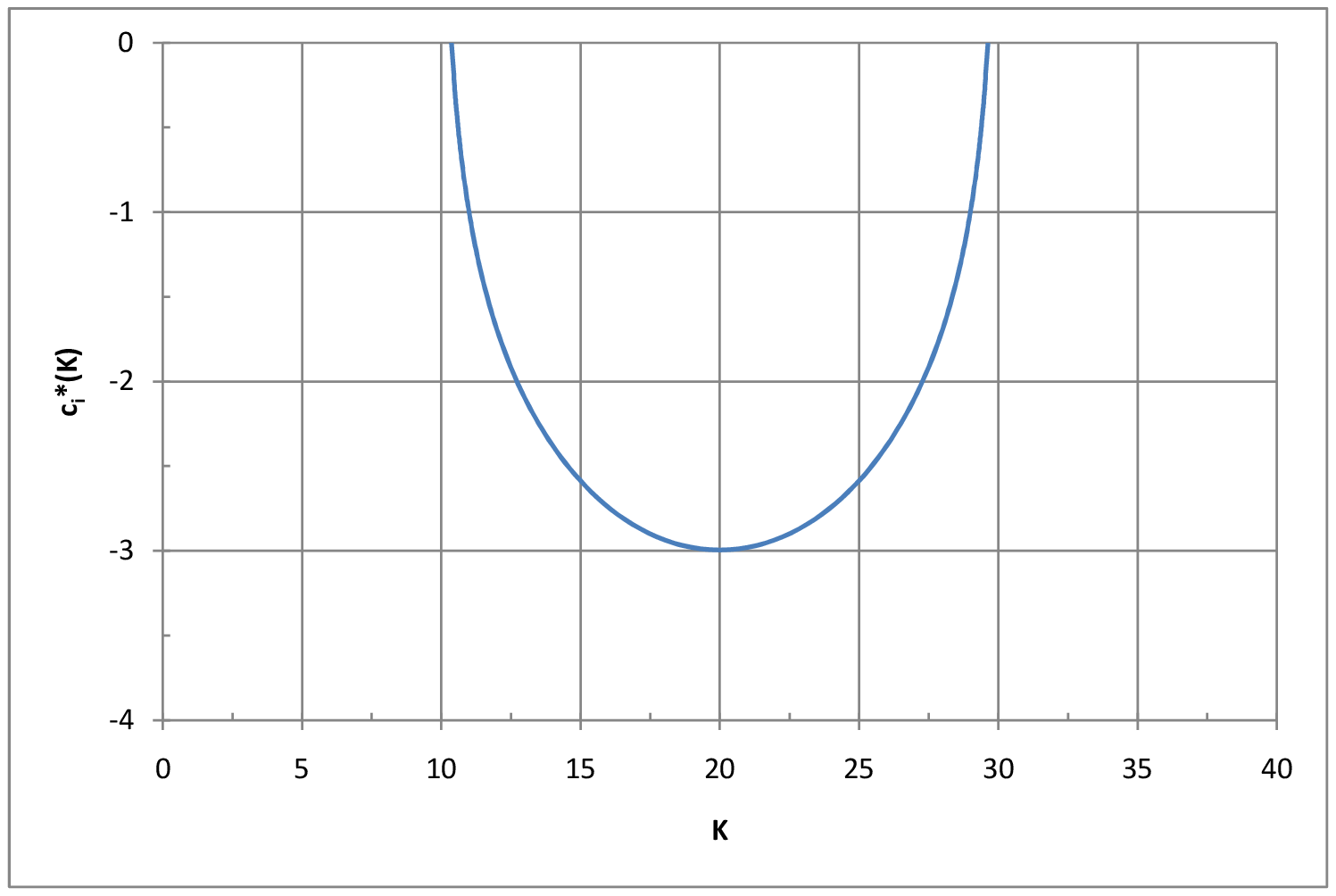}
\label{Fig:1d}
}
\caption{Graphs of $c_i$, $c_i'$, $[c_i^*]'$ and $c_i^*$ when $K_i = 10$ and $K_{i+1} = 30$.}
\label{Fig:1}
\end{figure}

The functions $c_1, ..., c_n$ are differentiable and strictly convex, and we describe their Legendre transforms.
Figure \ref{Fig:1} shows the graphs of $c_i$, $c_i'$, $[c_i^*]'$ and $c_i^*$ when $K_i = 10$ and $K_{i+1} = 30$.
In \cite{NeriSchneider2009} we have shown that $c_i'$ is continuously differentiable and strictly increasing,
and that $c_i'(\beta)$ goes to $K_i$ and $K_{i+1}$ when $\beta$ goes to $+\infty$ and $-\infty$, respectively.
For $i < n$, we have $c_i'(0) = (K_i + K_{i+1})/2$, and from Proposition \ref{PropertiesLegendreTransform} \ref{Derivatives}--\ref{Twice} it follows that
$(K_i + K_{i+1})/2$ is a root of $[c_i^*]'$ and a minimum of $c_i^*$.

From Proposition \ref{PropertiesLegendreTransform} \ref{Young}--\ref{Derivatives}, we have the equivalences
\begin{equation}
\label{CPlusCStar}
[c_i^*]'(K) = \beta \quad\Longleftrightarrow\quad
K = c_i'(\beta) \quad\Longleftrightarrow\quad
c_i(\beta) + c_i^*(K) = \beta K.
\end{equation}

Finally, recall from \cite{NeriSchneider2009} the following definitions and results.
Under no-arbitrage conditions \eqref{NoArbitrageDigitals}, that is, for all $\tilde{D}\in\Omega$, there is a unique density $g_{\tilde{D}}$ which maximises entropy and matches call \eqref{CallConstraints} and digital \eqref{DigitalConstraints} prices.
Moreover, for all $i=0, ..., n$ we have:
\begin{itemize}
\item The probability of $S(T)$ being inside a bucket $[K_i, K_{i+1}[$ is
\begin{equation}
\label{p_i}
p_i := {\mathbb P}^{g_{\tilde{D}}}[K_i \leq S(T) < K_{i+1}] =
\int_{K_i}^{K_{i+1}} g_{\tilde{D}}(x) dx = \tilde{D}_i - \tilde{D}_{i+1}.
\end{equation}
\item The conditional expectation of $S(T)$ given that it is inside the bucket is
\begin{equation}
\label{K_i_bar}
\bar{K}_i := {\mathbb E}^{g_{\tilde{D}}}[S(T) | K_i \leq S(T) < K_{i+1}] =
\frac{1}{p_i} \int_{K_i}^{K_{i+1}} x g_{\tilde{D}}(x) dx =
\frac{(\tilde{C}_i + K_i \tilde{D}_i) - (\tilde{C}_{i+1} + K_{i+1} \tilde{D}_{i+1})}{p_i}.
\end{equation}
\end{itemize}

Furthermore, on $[K_i, K_{i+1}[$ the density $g_{\tilde{D}}$ has the form $g_{\tilde{D}}(x) = \alpha_i e^{\beta_i x}$ where
\begin{align}
\alpha_i &= p_i e^{-c_i(\beta_i)},
\label{alpha_i}
\\
c_i'(\beta_i) &= \bar{K}_i.
\label{beta_i}
\end{align}

For notational convenience again, in \eqref{K_i_bar}, we make the convention that $K_{n+1}\tilde{D}_{n+1}:=0$.

Remark that in order for \eqref{p_i} and \eqref{K_i_bar} to hold, it is necessary that $p_i>0$ and $K_i<\bar{K}_i<K_{i+1}$ for all $i=0, ..., n$. These conditions are equivalent to $\eqref{NoArbitrageDigitals}$.

When computing $\alpha_i$ and $\beta_i$, equation \eqref{beta_i} is the only one that cannot be solved explicitly, since we do not know the expression of $[c_i']^{-1}=[c_i^*]'$.
However, in practice, it can be solved very quickly with a one-dimensional Newton-Raphson rootfinder, since we have the derivative $c_i''$ analytically.
Typically, it takes about three steps to obtain a very precise result, since the two functions are very well behaved.

\section{Maximum Entropy in Terms of Market Data}

Using the Legendre transform, we now show that for any $\tilde{D}\in\Omega$ the entropy of  $g_{\tilde{D}}\in\mathcal{G}$ can be expressed directly in terms of market data with no need to find the piecewise-exponential form of $g_{\tilde{D}}$.

\begin{theorem}
\label{H(D)}
For all $\tilde{D}\in\Omega$ the entropy of $g_{\tilde{D}}$ can be expressed as
\begin{equation*}
H(g_{\tilde{D}}) = -\sum_{i=0}^n p_i \ln p_i - \sum_{i=0}^n p_i c_i^*(\bar{K}_i),
\end{equation*}
where $p_i$ and $\bar{K}_i$ are given purely in terms of option prices by \eqref{p_i} and \eqref{K_i_bar}, for all $i=0, ..., n$.
\end{theorem}
\begin{proof}
Let $\tilde{D}\in\Omega$.
For each $i=0, ..., n$, let $p_i$, $\bar{K}_i$, $\alpha_i$ and $\beta_i$ be given by equations \eqref{p_i} and \eqref{K_i_bar}, \eqref{alpha_i} and \eqref{beta_i}, respectively.

Fix $i\in\{1, ..., n\}$ and let $H_i(g_{\tilde{D}})$ denote the entropy of $g_{\tilde{D}}$ over the $i$-th bucket, that is,
\begin{equation*}
H_i(g_{\tilde{D}}) := -\int_{K_i}^{K_{i+1}} g_{\tilde{D}}(x) \ln g_{\tilde{D}}(x) dx =
-\ln \alpha_i \int_{K_i}^{K_{i+1}} g_{\tilde{D}}(x) dx - \beta_i \int_{K_i}^{K_{i+1}} x g_{\tilde{D}}(x) dx.
\end{equation*}
From \eqref{p_i} and \eqref{K_i_bar} it follows that
\begin{equation*}
H_i(g_{\tilde{D}}) = -p_i(\ln\alpha_i + \beta_i\bar{K}_i).
\end{equation*}
Using \eqref{alpha_i},  we get
\begin{equation*}
H_i(g_{\tilde{D}}) = -p_i(\ln p_i - c_i(\beta_i) + \beta_i\bar{K}_i).
\end{equation*}
Finally, \eqref{beta_i} and \eqref{CPlusCStar} yield
\begin{equation*}
H_i(g_{\tilde{D}}) = -p_i \ln p_i - p_i c_i^*(\bar{K}_i).
\end{equation*}
Adding up over $i=0, ..., n$ yields
\begin{equation*}
H(g_{\tilde{D}}) = \sum_{i=0}^n H_i(g_{\tilde{D}}) = -\sum_{i=0}^n p_i \ln p_i - \sum_{i=0}^n p_i c_i^*(\bar{K}_i).
\end{equation*}

\vspace{-2\baselineskip}
\qed
\end{proof}

The expression of $H(g_{\tilde{D}})$ given in Theorem \ref{H(D)} can be split into discrete and continuous parts:
\begin{alignat}{3}
H^d(g_{\tilde{D}}) &:= \sum_{i=0}^n H^d_i(g_{\tilde{D}}), & \quad & \text{where } H^d_i(g_{\tilde{D}}) := -p_i \ln p_i, & \quad & \forall i = 0, ..., n,
\label{H_discrete}
\\
H^c(g_{\tilde{D}}) &:= \sum_{i=0}^n H^c_i(g_{\tilde{D}}), & \quad & \text{where } H^c_i(g_{\tilde{D}}) := -p_i c_i^*(\bar{K}_i), & \quad & \forall i = 0, ..., n.
\label{H_continuous}
\end{alignat}
The term $H^d(g_{\tilde{D}})$ can be seen as a Shannon type discrete entropy, whose maximum is attained when all $p_i$ are equal, {\em i.e.} when all buckets have equal probability, whereas each $H^c_i(g_{\tilde{D}})$, for $i < n$, attains its maximum when $\bar{K}_i$ lies exactly in the middle of $K_i$ and $K_{i+1}$, {\em i.e.} $\bar{K}_i = (K_i + K_{i+1})/2$.

\subsection{The Sensitivity of Maximum Entropy with Respect to Digital Prices}

With fixed strikes and call prices, under the light of Theorem \ref{H(D)}, $H$, $H^d$ and $H^c$ can be seen as functions of digital prices only.
Now we show how $H$ is affected by digital price changes.

\begin{theorem}
\label{H'(D)}
As a function of digital prices, $H:\Omega\rightarrow\mathbb{R}$ is differentiable and, for all $\tilde{D}\in\Omega$, we have
\begin{equation*}
\frac{\partial H}{\partial \tilde{D}_i}(\tilde{D}) = \ln g_{\tilde{D}}(K_i-) - \ln g_{\tilde{D}}(K_i+), \quad \forall i = 1, ..., n,
\end{equation*}
where
\begin{equation*}
g_{\tilde{D}}(K_i-) := \lim_{x\rightarrow K_i^-} g_{\tilde{D}}(x) = \alpha_{i-1} e^{\beta_{i-1} K_i}
\quad \text{and} \quad
g_{\tilde{D}}(K_i+) := \lim_{x\rightarrow K_i^+} g_{\tilde{D}}(x) = \alpha_i e^{\beta_i K_i},
\end{equation*}
with $\alpha_i$ and $\beta_i$ given by \eqref{alpha_i} and \eqref{beta_i}, respectively, for all $i=0, ..., n$.
\end{theorem}
\begin{proof}
Let $\tilde{D}\in\Omega$.
For each $i=0, ..., n$, let $p_i$, $\bar{K}_i$, $\alpha_i$ and $\beta_i$ be given by equations \eqref{p_i} and \eqref{K_i_bar}, \eqref{alpha_i} and \eqref{beta_i}, respectively.

Fix $i\in\{1, ..., n\}$.
Throughout the proof, we use the following consequences of \eqref{p_i} and \eqref{K_i_bar}:
\begin{equation*}
\dfrac{\partial p_i}{\partial\tilde{D}_i} = 1,
\quad
\dfrac{\partial p_{i-1}}{\partial\tilde{D}_i} = -1,
\quad
\dfrac{\partial\bar{K}_{i-1}}{\partial p_{i-1}} = \dfrac{K_{i-1} - \bar{K}_{i-1}}{p_{i-1}}
\quad\text{and}\quad
\dfrac{\partial\bar{K}_i}{\partial p_i} = \dfrac{K_i - \bar{K}_i}{p_i}.
\end{equation*}

Starting with the discrete part, from \eqref{H_discrete} we have
\begin{align}
\frac{\partial H^d}{\partial \tilde{D}_i}(\tilde{D})
&=
\frac{\partial}{\partial \tilde{D}_i}\Big[-p_{i-1} \ln p_{i-1} - p_i \ln p_i \Big]
=
\frac{\partial p_{i-1}}{\partial \tilde{D}_i}\Big(-\ln p_{i-1} - 1 \Big)
+ \frac{\partial p_i}{\partial \tilde{D}_i} \Big(-\ln p_i - 1 \Big)
\nonumber
\\
&=
\ln p_{i-1} - \ln p_i.
\label{dHddDi}
\end{align}

For the continuous part, from \eqref{H_continuous} we have
\begin{align*}
\frac{\partial H^c}{\partial\tilde{D}_i}(\tilde{D})
&= \frac{\partial}{\partial \tilde{D}_i} \Big[ -p_{i-1} c_{i-1}^*(\bar{K}_{i-1}) - p_i c_i^*(\bar{K}_i) \Big] \\
&= \frac{\partial p_{i-1}}{\partial\tilde{D}_i} \Big(-c_{i-1}^*(\bar{K}_{i-1}) -
p_{i-1} [c_{i-1}^*]'(\bar{K}_{i-1}) \frac{\partial\bar{K}_{i-1}}{\partial p_{i-1}} \Big)
+
\frac{\partial p_i}{\partial\tilde{D}_i} \Big(-c_{i}^*(\bar{K}_i) -
p_i [c_i^*]'(\bar{K}_i) \frac{\partial \bar{K}_i}{\partial p_i} \Big) \\
&= c_{i-1}^*(\bar{K}_{i-1}) +
p_{i-1} [c_{i-1}^*]'(\bar{K}_{i-1}) \left(\frac{\bar{K}_{i-1} - K_i}{p_{i-1}}\right)
- c_i^*(\bar{K}_i) - p_i [c_i^*]'(\bar{K}_i)\left(\frac{K_i - \bar{K}_i}{p_i}\right).
\end{align*}
Using \eqref{CPlusCStar} and \eqref{beta_i} yields $[c_{i-1}^*]'(\bar{K}_{i-1}) =
\beta_{i-1}$ and $[c_i^*]'(\bar{K}_i) = \beta_i$. Hence,
\begin{align*}
\frac{\partial H^c}{\partial \tilde{D}_i}(\tilde{D})
&=
c_{i-1}^*(\bar{K}_{i-1}) - \beta_{i-1} (\bar{K}_{i-1} - K_i) -
c_i^*(\bar{K}_i) - \beta_i (K_i - \bar{K}_i)
\\
&=
c_{i-1}^*(\bar{K}_{i-1}) - \beta_{i-1} \bar{K}_{i-1} + \beta_{i-1} K_i
- c_i^*(\bar{K}_i) + \beta_i \bar{K}_i - \beta_i K_i.
\end{align*}
Again, using \eqref{CPlusCStar} we get
\begin{equation}
\frac{\partial H^c}{\partial \tilde{D}_i}(\tilde{D})
=
- c_{i-1}(\beta_{i-1}) + \beta_{i-1} K_i + c_i(\beta_i) - \beta_i K_i.
\label{dHcdDi}
\end{equation}

Putting \eqref{dHddDi} and \eqref{dHcdDi} together and using \eqref{alpha_i} we obtain
\begin{align*}
\frac{\partial H}{\partial \tilde{D}_i}(\tilde{D}) &= \frac{\partial H^d}{\partial \tilde{D}_i}(\tilde{D}) + \frac{\partial H^c}{\partial \tilde{D}_i}(\tilde{D})
= \ln p_{i-1} - c_{i-1}(\beta_{i-1}) + \beta_{i-1} K_i - \ln p_i + c_i(\beta_i) - \beta_i K_i
\\
&=
\ln \alpha_{i-1} + \beta_{i-1} K_i - \ln \alpha_i - \beta_i K_i
= \ln (\alpha_{i-1} e^{\beta_{i-1} K_i}) - \ln (\alpha_i e^{\beta_i K_i})
\\
&=
\ln g_{\tilde{D}}(K_i-) - \ln g_{\tilde{D}}(K_i+).
\end{align*}

\vspace{-2\baselineskip}
\qed
\end{proof}

\begin{corollary}
\label{Critical=Continuous}
For any $\tilde{D}\in\Omega$, the density $g_{\tilde{D}}\in\mathcal{G}$ is continuous if, and only if, $\tilde{D}$ is a critical point of $H$, that is,
\label{Corollary1}
\begin{equation*}
\frac{\partial H}{\partial \tilde{D}_i}(\tilde{D}) = 0 \quad \forall i = 1, ..., n.
\end{equation*}
\end{corollary}
\begin{proof}
This is a direct consequence of Theorem \ref{H'(D)} and the fact that in the interior of each bucket $g_{\tilde{D}}$ is exponential.
\qed
\end{proof}

\section{The Buchen-Kelly Density is the only Continuous MED}

We know that the Buchen-Kelly density is continuous and has the greatest entropy of all densities in $\mathcal{G}$.
In agreement with Corollary \ref{Critical=Continuous}, digital prices implied by the Buchen-Kelly density correspond to a point of maximum entropy and, thus, a critical point of $H$ in $\Omega$.
Actually, the digital prices implied by the Buchen-Kelly density form the only critical point of $H$ in $\Omega$.
Indeed, $H$ is strictly concave as we shall see in Proposition \ref{Concavity}.

Now we state and prove a lemma which we need for our algorithm and for the proof of Proposition \ref{Concavity}.

\begin{lemma}
\label{H''(D)}
The function $H:\Omega\rightarrow\mathbb{R}$ is twice differentiable and its Hessian matrix at any $\tilde{D}\in\Omega$ is symmetric and tridiagonal with entries given by
\begin{align*}
\dfrac{\partial^2 H}{\partial \tilde{D}_i^2}(\tilde{D})
&=
-\dfrac{1}{p_{i-1}}
-\dfrac{1}{p_i}
-\dfrac{(K_i - \bar{K}_{i-1})^2}{p_{i-1} c_{i-1}''(\beta_{i-1})}
-\dfrac{(\bar{K}_i - K_i)^2}{p_i c_i''(\beta_i)}, \quad &
\forall i &= 1, ..., n, \\
\\
\dfrac{\partial^2 H}{\partial \tilde{D}_i \partial \tilde{D}_{i+1}}(\tilde{D}) &=
\dfrac{1}{p_i}
-\dfrac{(\bar{K}_i - K_i)(K_{i+1}-\bar{K}_i)}{p_i c_i''(\beta_i)}, \quad &
\forall i &= 1, ..., n - 1,
\end{align*}
where $p_i$, $\bar{K}_i$ and $\beta_i$ are given by \eqref{p_i}, \eqref{K_i_bar} and \eqref{beta_i}, respectively, for all $i = 0, ..., n$.
\end{lemma}
\begin{proof}
Let $\tilde{D}\in\Omega$.
For each $i=0, ..., n$, let $p_i$, $\bar{K}_i$ and $\beta_i$ be given by equations \eqref{p_i}, \eqref{K_i_bar} and \eqref{beta_i}, respectively.

The fact that $H''(\tilde{D})$ is tridiagonal and symmetric is clear from the expression of $H$'s partial derivatives given in Theorem \ref{H'(D)}.
Actually, the discrete part $[H^d]''(\tilde{D})$ and continuous part $[H^c]''(\tilde{D})$ of $H''(\tilde{D})$ are themselves symmetric and tridiagonal.

Fix $i\in\{1, ..., n\}$.
We begin by calculating the second derivatives of $H^d$ starting from \eqref{dHddDi}.
For the diagonal entry we have
\begin{equation}
\label{HessianDiagonalDiscrete}
\frac{\partial^2 H^d}{\partial \tilde{D}_i^2}(\tilde{D})
= \frac{\partial}{\partial \tilde{D}_i} \Big[\ln p_{i-1} - \ln p_i \Big]
= \frac{1}{p_{i-1}} \frac{\partial p_{i-1}}{\partial \tilde{D}_i}
- \frac{1}{p_i} \frac{\partial p_i}{\partial \tilde{D}_i}
= - \frac{1}{p_{i-1}} -\frac{1}{p_i},
\end{equation}
and, if $i<n$, for the off-diagonal entry we have
\begin{equation}
\label{HessianOffDiagonalDiscrete}
\frac{\partial^2 H^d}{\partial \tilde{D}_i \partial \tilde{D}_{i+1}}(\tilde{D})
= -\frac{\partial}{\partial \tilde{D}_{i+1}}\Big[\ln p_i\Big]
= -\frac{1}{p_i} \frac{\partial p_i}{\partial \tilde{D}_{i+1}} = \frac{1}{p_i}.
\end{equation}

Next, for the continuous part, starting from \eqref{dHcdDi} we get
\begin{align}
\frac{\partial^2 H^c}{\partial \tilde{D}_i^2}(\tilde{D}) &= \frac{\partial}{\partial\tilde{D}_i}
\Big[- c_{i-1}(\beta_{i-1}) + \beta_{i-1} K_i + c_i(\beta_i) - \beta_i K_i\Big]
\nonumber
\\
&=
- c_{i-1}'(\beta_{i-1}) \frac{\partial \beta_{i-1}}{\partial \tilde{D}_i}
+ \frac{\partial \beta_{i-1}}{\partial \tilde{D}_i} K_i
+ c_i'(\beta_i) \frac{\partial \beta_i}{\partial \tilde{D}_i}
- \frac{\partial \beta_i}{\partial \tilde{D}_i} K_i
\nonumber
\\
&=
(K_i - \bar{K}_{i-1}) \frac{\partial \beta_{i-1}}{\partial \tilde{D}_i}
+
(\bar{K}_i - K_i) \frac{\partial \beta_i}{\partial \tilde{D}_i}
\nonumber
\\
&=
\label{HessianDiagonalContinuous}
- \frac{(K_i - \bar{K}_{i-1})^2}{p_{i-1} c_{i-1}''(\beta_{i-1})}
- \frac{(\bar{K}_i - K_i)^2}{p_i c_i''(\beta_i)},
\end{align}
and, if $i<n$,
\begin{equation}
\label{HessianOffDiagonalContinuous}
\frac{\partial^2 H^c}{\partial \tilde{D}_i \partial \tilde{D}_{i+1}}(\tilde{D})
= \frac{\partial}{\partial \tilde{D}_{i+1}}\Big[ c_i(\beta_i) - \beta_i K_i \Big]
= (c_i'(\beta_i) -K_i) \frac{\partial \beta_i}{\partial \tilde{D}_{i+1}}
= -\frac{(\bar{K}_i - K_i)(K_{i+1} - \bar{K}_i)}{p_i c_i''(\beta_i)}.
\end{equation}
In the last steps, we used
\begin{equation*}
\frac{\partial \beta_{i-1}}{\partial \tilde{D}_i} = -\frac{(K_i - \bar{K}_{i-1})}{p_{i-1} c_{i-1}''(\beta_{i-1})},
\quad
\frac{\partial \beta_i}{\partial \tilde{D}_i} = -\frac{(\bar{K}_i - K_i)}{p_i c_i''(\beta_i)}
\quad\text{and}\quad
\frac{\partial \beta_i}{\partial \tilde{D}_{i+1}} = -\frac{(K_{i+1} - \bar{K}_i)}{p_i c_i''(\beta_i)},
\end{equation*}
which are obtained by differentiating equation $c_{i-1}'(\beta_{i-1}) = \bar{K}_{i-1}$ w.r.t. $\tilde{D}_i$ and equation $c_i'(\beta_i) = \bar{K}_i$ w.r.t. $\tilde{D}_i$ and $\tilde{D}_{i+1}$.
\qed
\end{proof}

\begin{proposition}
\label{Concavity}
The function $H:\Omega\rightarrow\mathbb{R}$ is strictly concave.
\end{proposition}
\begin{proof}
Let $\tilde{D}\in\Omega$.
It is enough to show that $H''(\tilde{D})$ is strictly negative definite.
To accomplish this, we shall prove that $[H^d]''(\tilde{D})$ is strictly negative definite and $[H^c]''(\tilde{D})$ is negative definite.

For each $i=0, ..., n$, let $p_i$, $\bar{K}_i$ and $\beta_i$ be given by equations \eqref{p_i}, \eqref{K_i_bar} and \eqref{beta_i}, respectively.

Recall that for an $n\times n$ triangular symmetric matrix $A=(a_{i,j})$ we have
\begin{equation}
\label{TridiagonalQuadratic}
\left<Ad,d\right> =
\sum_{i=1}^n a_{i,i}d_i^2 + 2\sum_{i=1}^{n-1} a_{i,i+1}d_id_{i+1} \quad \forall
d=(d_1, ..., d_n)\in\mathbb{R}^n.
\end{equation}

Apply this relation to $A=H^d(\tilde{D})$ and $d\in\mathbb{R}^n$, with $d\ne 0$, and use \eqref{HessianDiagonalDiscrete} and \eqref{HessianOffDiagonalDiscrete} to get
\begin{align}
\left<[H^d]''(\tilde{D})d, d\right>
&= - \sum_{i=1}^n \left(\frac1{p_{i-1}} + \frac1{p_i}\right)d_i^2
   + 2\sum_{i=1}^{n-1} \frac{d_id_{i+1}}{p_i}
\nonumber \\
&= - \frac{d_1^2}{p_0}
   - \sum_{i=1}^{n-1} \frac{d_i^2 + d_{i+1}^2}{p_i}
   - \frac{d_n^2}{p_n}
   + 2\sum_{i=1}^{n-1} \frac{d_id_{i+1}}{p_i}
\nonumber \\
&= - \frac{d_1^2}{p_0}
   - \sum_{i=1}^{n-1} \frac{(d_i - d_{i+1})^2}{p_i}
   - \frac{d_n^2}{p_n}
\label{HdConcave} \\
&< 0. \nonumber
\end{align}

Now, apply \eqref{TridiagonalQuadratic} to $A=H^c(\tilde{D})$ and $d\in\mathbb{R}^n$ and use \eqref{HessianDiagonalContinuous} and \eqref{HessianOffDiagonalContinuous} to obtain
\begin{align*}
\left<[H^c]''(\tilde{D})d, d\right>
&=
- \sum_{i=1}^n
\left[
  \frac{(K_i - \bar{K}_{i-1})^2}{p_{i-1}c_{i-1}''(\beta_{i-1})}
  +
  \frac{(\bar{K}_i - K_i)^2}{p_ic_i''(\beta_i)}
\right]d_i^2
- 2\sum_{i=1}^{n-1}
\frac{(\bar{K}_i - K_i)(K_{i+1} - \bar{K}_i)d_id_{i+1}}{p_ic''(\beta_i)}
\\
&=
- \frac{(K_1 - \bar{K}_0)^2d_1^2}{p_0c_0''(\beta_0)}
- \sum_{i=1}^{n-1}
  \frac{(\bar{K}_i - K_i)^2d_i^2 + (K_{i+1} - \bar{K}_i)^2d_{i+1}^2}{p_ic_i''(\beta_i)}
- \frac{(\bar{K}_n - K_n)^2d_n^2}{p_nc_n''(\beta_n)}
\\
& \qquad\qquad\qquad\qquad
- 2\sum_{i=1}^{n-1}
\frac{(\bar{K}_i - K_i)d_i(K_{i+1} - \bar{K}_i)d_{i+1}}{p_ic''(\beta_i)}
\\
&=
- \frac{(K_1 - \bar{K}_0)^2d_1^2}{p_0c_0''(\beta_0)}
- \sum_{i=1}^{n-1}
  \frac{\left[(\bar{K}_i - K_i)d_i + (K_{i+1} - \bar{K}_i)d_{i+1}\right]^2}{p_ic_i''(\beta_i)}
- \frac{(\bar{K}_n - K_n)^2d_n^2}{p_nc_n''(\beta_n)}
\\
&\le 0.
\end{align*}
(Actually, with a small extra effort one gets strict inequality above for $d\ne 0$.)
\qed
\end{proof}

Remark that the last proof implies that $H''(\tilde{D})$ is invertible for all $\tilde{D}\in\Omega$.
This is important to assure that the Newton-Raphson step for maximizing $H$ is well defined.
We will come back to this point in Section \ref{sec:Algorithm}.

From Proposition \ref{Concavity} we immediately obtain the next result.
\begin{corollary}
\label{BKCharacterization}
The following statements are equivalent:
\begin{enumerate}
\item\label{BKDensity} $g_{\tilde{D}}\in\mathcal{G}$ is the Buchen-Kelly density.
\item\label{Critical} $\tilde{D}\in\Omega$ is a critical point of $H$.
\item\label{Continuous} $g_{\tilde{D}}\in\mathcal{G}$ is continuous.
\end{enumerate}
\end{corollary}
\begin{proof}
We already knew that \ref{BKDensity} $\Rightarrow$ \ref{Critical} $\Leftrightarrow$ \ref{Continuous},
and the last proposition fills the gap \ref{Critical} $\Rightarrow$ \ref{BKDensity}.
\qed
\end{proof}

\subsection{Some Remarks Regarding the Buchen-Kelly Density}
\label{sec:BKNotation}

We conclude this section by giving some simple but useful formulas for expressing the Buchen-Kelly density in terms of our $\alpha_i$'s and $\beta_i$'s, and vice versa.
The analytical formulas from Section 2.3 in \cite{NeriSchneider2009} can then be applied to densities in Buchen-Kelly form.

Buchen and Kelly label the calls $\tilde{C}_i$ observed at strikes $K_i$ in the market from $i=1, ...,m$.
The density in \cite{BuchenKelly1996} is explicitly given there in equation (8), for $c_i(x) = (x - K_i)^+$, by
\begin{equation}
\label{BuchenKellyDensity}
g(x) = \frac{1}{\mu} \exp \left( \sum_{i=1}^m \lambda_i (x - K_i)^+ \right), \quad \mu = \int_0^\infty \exp \left( \sum_{i=1}^m \lambda_i (x - K_i)^+ \right) dx.
\end{equation}
For $x \in [K_j, K_{j+1}[$, we can write this as
\begin{equation*}
g(x) = \frac{1}{\mu} \exp \left( -\sum_{i=1}^j \lambda_i K_i \right) \exp \left( \sum_{i=1}^j \lambda_i x \right).
\end{equation*}
Assume $K_1 = 0$. Then we obtain the following conversion formulas:
\begin{alignat*}{4}
\alpha_0 &= \frac{1}{\mu}, & \quad
\alpha_1 &= \frac{1}{\mu} e^{-\lambda_2 K_2}, & \quad
&..., & \quad
\alpha_{m-1} &= \frac{1}{\mu} e^{-\sum_{i=2}^m \lambda_i K_i},
\\
\beta_0 &= \lambda_1, &
\beta_1 &= \lambda_1 + \lambda_2, &
&..., &
\beta_{m-1} &= \sum_{i=1}^m \lambda_i.
\end{alignat*}

Obviously only a continuous density can be written in the Buchen-Kelly form. For such a density, we have
\begin{alignat*}{4}
\lambda_1 &= \beta_0, & \quad
\lambda_2 &= \beta_1 - \beta_0, & \quad
&..., & \quad
\lambda_m &= \beta_{m-1} - \beta_{m-2},
\\
\mu &= \frac{1}{\alpha_0}.
\end{alignat*}

\section{An Algorithm to Find the Buchen-Kelly Density}
\label{sec:Algorithm}

Corollary \ref{BKCharacterization} says that if $\tilde{D}\in\Omega$ is a root of $H'$, then $g_{\tilde{D}}$ is the Buchen-Kelly density.
Since the Hessian matrix of $H$ is known analytically, the Newton-Raphson method is a possible choice of algorithm to find the root of $H'$ numerically,
and consequently the Buchen-Kelly density.

The choice of the Newton-Raphson method for root-finding problems is not without
its concerns.
Nevertheless, in this section we shall see that for our particular problem these concerns vanish if we combine a {\it pure} Newton-Raphson method with an initial phase of backtracking line search.
This is the {\it damped} Newton method presented in \cite{BoydVandenberghe2004}, page 487.

First of all, there is the issue of finding the {\em wrong} root when there is more than one.
Of course, this cannot happen here since uniqueness holds.

Then there is the high cost of computing and storing the, normally, $n^2$ entries of the Hessian matrix.
In our case, this matrix is symmetric and tridiagonal (Lemma \ref{H''(D)}) and, thus, the cost applies only to $2n-1$ entries.

Each step of the Newton-Raphson method for finding the root of $H'$ entails solving, for $d\in\mathbb{R}^n$, the linear system $H''(\tilde{D}) d = H'(\tilde{D})$, where $\tilde{D}\in\Omega$ is the current guess for the solution.
Therefore, the method fails if $H''(\tilde{D})$ is not invertible.
However, this cannot happen since  we have $\langle H''(\tilde{D})d, d\rangle < 0$ if $d\ne 0$ (Proposition \ref{StrongConcavity}).

Again, solving the $n\times n$ linear system above might be very costly in general.
Here, since $H'(\tilde{D})$ is tridiagonal, the system can be efficiently solved in $\mathcal{O}(n)$ operations using, for example the routine
{\it tridag} from \cite{PVTF2002}.
Regarding the stability of this routine, we quote \cite{PVTF2002}, page 57:
\begin{quote}
The tridiagonal algorithm is the rare case of an algorithm that, in practice, is more robust than theory says it should be.
\end{quote}
Instability of the linear system is then unlikely but if it does occur, we still have room for improvement. There is a well known link between instability and $H''(\tilde{D})$ having a large condition number. The latter is connected to the eccentricity of level curves of $H$. Hence the geometry of the $n$-dimensional rectangle $\Omega$ plays an important role. By making $\Omega$ less eccentric through affine transformations, one can decrease the condition number of $H''(\tilde{D})$, thus improving stability of the linear system. This does not affect the Newton-Raphson algorithm since it is invariant by affine transformations \cite{BoydVandenberghe2004}.

It is worth mentioning that linear tridiagonal systems arise in a broad range of problems.
For this reason they have been extensively studied and several algorithms (including parallel ones) have been proposed \cite{Higham1990}, \cite{Lewis1982}, \cite{LinChung1990}.

The last concern is that a {\it pure} Newton-Raphson method does not guarantee that each guess stays in $\Omega$, the domain of $H$.
In combination with the backtracking line search this issue becomes just an implementation detail as explained in \cite{BoydVandenberghe2004}, page 465.

The remainder of this section presents results that allow us to follow the convergence analysis of the algorithm presented in \cite{BoydVandenberghe2004}.

The next result strengthens Proposition \ref{Concavity}.
\begin{proposition}
\label{StrongConcavity}
The function $H:\Omega\rightarrow\mathbb{R}$ is strongly concave, that is, there exists an $m>0$ such that
\begin{equation*}
\left<H''(\tilde{D}) d, d\right> \le -m\|d\|^2, \quad \forall \tilde{D}\in\Omega \quad \text{and} \quad \forall d\in\mathbb{R}^n.
\end{equation*}
Furthermore, the constant $m$ can be taken as $4\sin^2(\pi/(2n+2))$.
\end{proposition}
\begin{proof}
In the proof of Proposition \ref{Concavity} we showed that both $H^d$ and $H^c$ satisfy a relationship similar to the one we aim to prove but with $m=0$.
Strengthening this relationship for $H^d$ to $m>0$ is enough to finish the proof.

Let $\tilde{D}\in\Omega$ and $d\in\mathbb{R}^n$.
From \eqref{HdConcave}, and using that $p_i\le 1$ for all $i=0, ..., n,$ we get
\begin{align*}
\left<[H^d]''(\tilde{D})d, d\right>
&\le
- d_1^2 -\sum_{i=1}^{n-1} (d_i - d_{i+1})^2 - d_n^2
=
- d_1^2 - \sum_{i=1}^{n-1} (d_i^2 + d_{i+1}^2) + 2\sum_{i=1}^{n-1} d_id_{i+1} - d_n^2
\\
&=
- 2\sum_{i=1}^n d_i^2 + 2\sum_{i=1}^{n-1} d_id_{i+1} = - \left<Ad, d\right>,
\end{align*}
where $A$ is the $n\times n$ tridiagonal symmetric matrix with $2$'s in the diagonal and $-1$'s in the off-diagonals.
This is a well known matrix that arises in the finite difference discretisation of the heat equation.
Its eigenvalues are $4\sin^2(k\pi/(2n + 2))$ for $k=1, ..., n$ \cite{WeidemanTrefethen1988}.
Therefore the result holds if we take $m=4\sin^2(\pi/(2n + 2))$, the smallest eigenvalue of $A$.
\qed
\end{proof}

Since $H$ is strongly concave, we can apply the arguments of \cite{BoydVandenberghe2004}, Section 9.1.2, which we repeat in the next proposition for the sake of completeness.
This proposition gives estimates on how far an element $\tilde{D}\in\Omega$ is from the Buchen-Kelly digitals and how much $H(\tilde{D})$ is below the maximum.
Both estimates are given in terms of $\|H'(\tilde{D})\|$.
Since $H'(\tilde{D})$ is easily computed (Proposition \ref{H'(D)}) at each Newton-Raphson step, these estimates can be used to define a stopping criterion.

\begin{proposition}
\label{Estimates}
Let $\hat{D}\in\Omega$ be such that $g_{\hat{D}}\in\mathcal{G}$ is the Buchen-Kelly density and let $m>0$ be as in Proposition \ref{StrongConcavity}.
Then, for all $\tilde{D}\in\Omega$ we have
\begin{equation*}
H(\hat{D}) - H(\tilde{D}) \le \frac1{2m} \|H'(\tilde{D})\|^2
\quad\text{and}\quad
\|\hat{D} - \tilde{D}\| \le \frac2{m} \|H'(\tilde{D})\|.
\end{equation*}
\end{proposition}
\begin{proof}
Let $\tilde{D}\in\Omega$ and set $d=\hat{D} - \tilde{D}$. The second-order Taylor expansion of $H$ around $\tilde{D}$ gives $\theta\in(0,1)$ such that
\begin{equation*}
H(\hat{D}) - H(\tilde{D})
=
\left<H'(\tilde{D}),d\right>
+ \frac12 \left<H''(\tilde{D} + \theta d)d, d\right>
\le
\left<H'(\tilde{D}), d\right>
- \frac{m}2 \|d\|^2.
\end{equation*}
Seen as a function of $d\in\mathbb{R}^n$, the rightmost term above is a concave quadratic function which reaches its maximum at $m^{-1}H'(\tilde{D})$.
Therefore, substituting $d$ by $m^{-1}H'(\tilde{D})$ yields our first result.

By the Cauchy-Schwarz inequality we have $\left<H'(\tilde{D}),d\right> \le \|H'(\tilde{D})\|\|d\|$.
Therefore,
\begin{equation*}
H(\hat{D}) - H(\tilde{D})
\le
\left<H'(\tilde{D}), d\right> - \frac{m}2 \|d\|^2
\le
\|d\|\left(\|H'(\tilde{D})\|- \frac{m}2 \|d\|\right).
\end{equation*}
Since $H(\hat{D}) - H(\tilde{D}) \ge 0$, so is the righthand side of the expression above, which implies the second result.
\qed
\end{proof}

We also have an estimate for the distance (in the sense of the $L^1(0,\infty)$ norm) between $g_{\tilde{D}}\in\mathcal{G}$ and the Buchen-Kelly density.
First we need to measure this distance in terms of the relative entropy.

For two strictly positive probability densities $f$ and $g$ over $[0, \infty)$, the relative entropy or {\it I-divergence} of $f$ with respect to $g$ is given by
\begin{equation*}
I(f\|g) = \int_0^\infty f(x)\ln\frac{f(x)}{g(x)} dx.
\end{equation*}
The relative entropy of $f$ with respect to $g$ is, in some ways, similar to a measure of the ``distance'' between these two densities \cite{Csiszar1975}.
For instance, $I(f\|g)\ge0$ with equality holding if, and only if, $f=g$.
However, the relative entropy is not a metric.

\begin{theorem}
\label{Th:L1Difference}
Let $g_{\hat{D}}\in\mathcal{G}$ be the Buchen-Kelly density. For all $\tilde{D}\in\Omega$, we have
\begin{equation*}
\|g_{\hat{D}} - g_{\tilde{D}}\|_{L^1} \le \frac{1}{\sqrt{m}}\|H'(\tilde{D})\|.
\end{equation*}
\end{theorem}
\begin{proof}
Recall that the Buchen-Kelly density has the form $g_{\hat{D}}(x) = \mu^{-1}\exp\left(\sum_{i=0}^n \lambda_i(x-K_i)^+\right)$ for all $x>0$ and some constants $\mu>0$, $\lambda_0, ..., \lambda_n$.
For any $\tilde{D}\in\Omega$ we have
\begin{equation*}
\int_0^\infty g_{\tilde{D}}\ln g_{\hat{D}} dx
=
- \int_0^\infty g_{\tilde{D}}\ln\mu dx
+ \sum_{i=0}^n \lambda_i \int_0^\infty(x-K_i)^+ g_{\tilde{D}} dx
=
- \ln\mu + \sum_{i=0}^n \lambda_i \tilde{C}_i.
\end{equation*}
In particular, the integral above does not depend on $\tilde{D}$.
Moreover, taking $\tilde{D}=\hat{D}$ gives that the integral is equal to $-H(\hat{D})$.
Therefore,
\begin{equation*}
I(g_{\tilde{D}}\|g_{\hat{D}})
=
\int_0^\infty g_{\tilde{D}}\ln\frac{g_{\tilde{D}}}{g_{\hat{D}}} dx
=
\int_0^\infty g_{\tilde{D}}\ln g_{\tilde{D}} dx
-
\int_0^\infty g_{\tilde{D}}\ln g_{\hat{D}} dx
=
H(\hat{D}) - H(\tilde{D}).
\end{equation*}

From the Csisz\'ar-Kullback inequality \cite{ArnoldMarkowichToscaniUnterreiter2000} and Proposition \ref{Estimates} we get
\begin{equation*}
\|g_{\hat{D}} - g_{\tilde{D}}\|_{L^1}^2
\le
2 I(g_{\tilde{D}}\|g_{\hat{D}})
=
2\left(H(\hat{D}) - H(\tilde{D})\right)
\le
\frac1{m} \|H'(\tilde{D})\|^2,
\end{equation*}
and the result follows.
\qed
\end{proof}

\subsection{Two Low-Dimensional Examples}

To see that usually we are dealing with a very smooth problem, consider the case where the forward $F = \tilde{C}_0$ and a call price $\tilde{C}_1$
at a strike $K_1$ are given. The digital $\tilde{D}_1$ (at the same strike $K_1$) can vary between the bounds imposed by \eqref{NoArbitrageDigitals}.

\begin{figure}[ht]
\centering
\subfigure[Graph of $H(\tilde{D}_1)$.] {
\includegraphics[width=.45\textwidth]{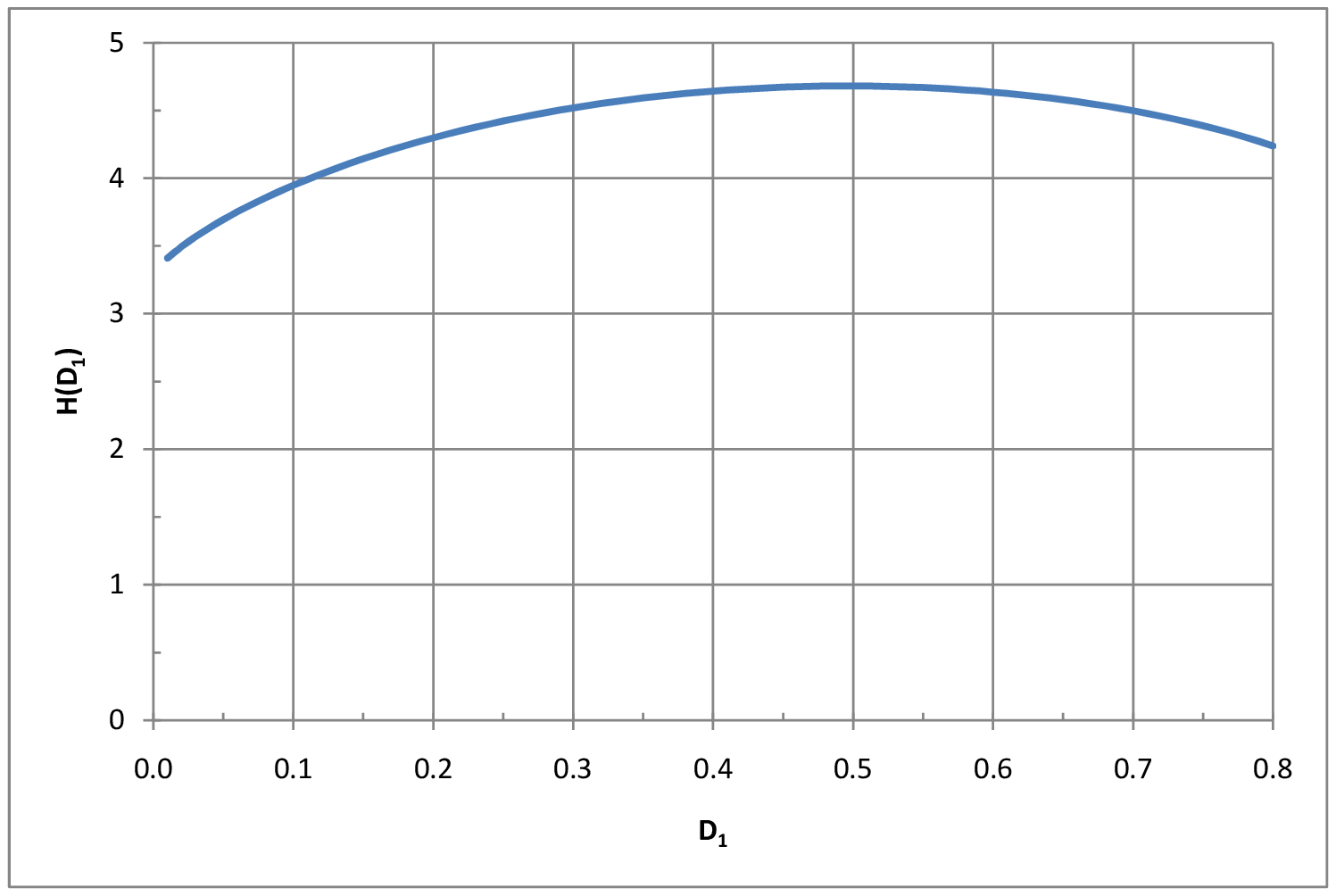}
\label{Fig:2a}
}
\subfigure[Graph of $H(\tilde{D}_1, \tilde{D}_2)$.] {
\includegraphics[width=.45\textwidth]{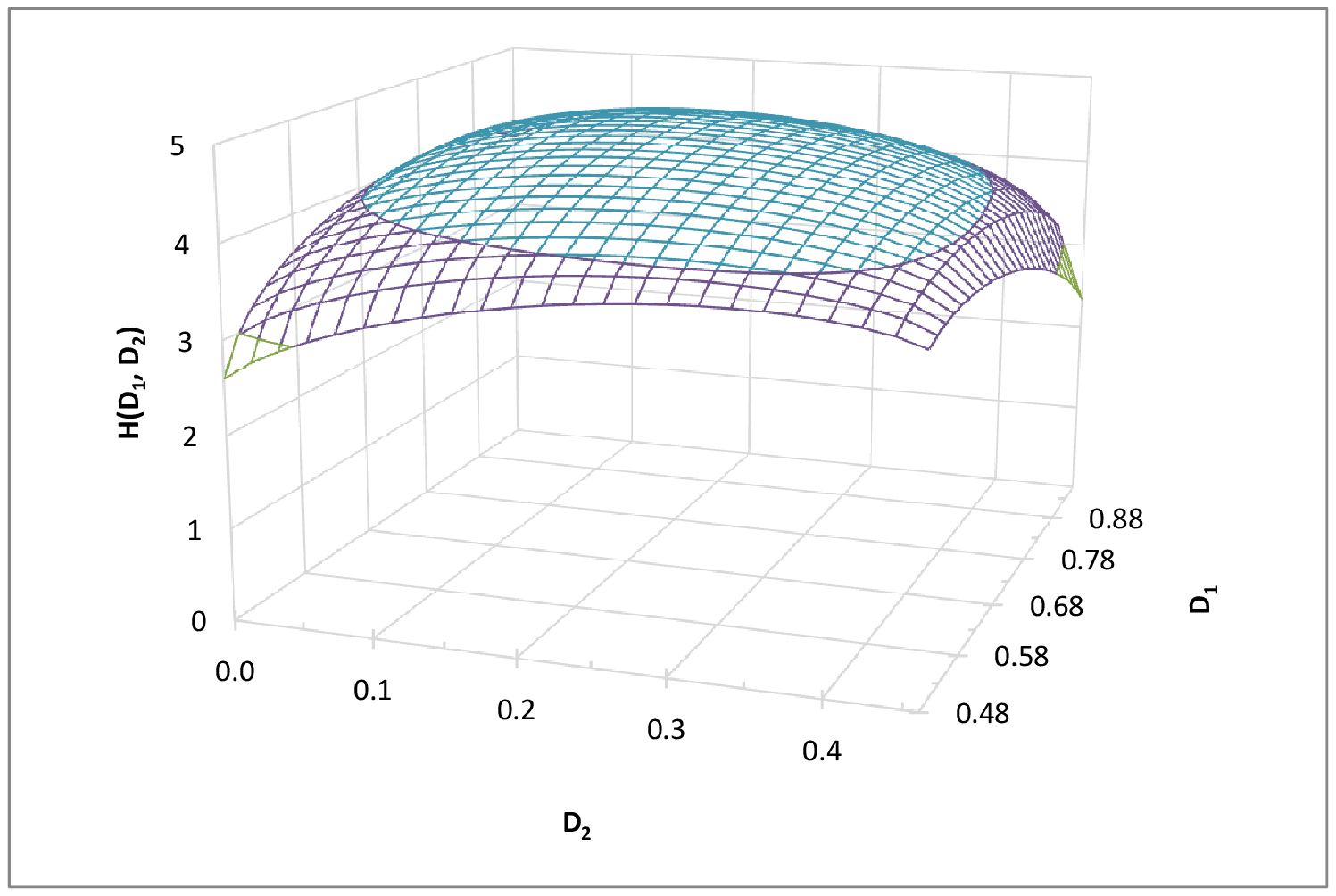}
\label{Fig:2b}
}
\caption{Graphs of $H(\tilde{D}_1)$ and $H(\tilde{D}_1, \tilde{D}_2)$.}
\label{Fig:2}
\end{figure}

In Figure \ref{Fig:2a} we have plotted the graph of $H$ for
\begin{alignat*}{2}
F &= 100.00,
\\
K_1 &= 100.00, & \quad \tilde{C}_1 &= 9.9477.
\end{alignat*}
By \eqref{NoArbitrageDigitals} we must have $\tilde{D}_1 \in \Omega = ]0, 0.9005[$.

$H$ attains its maximum of $4.6801$ at $\tilde{D}_1 = 0.4962$.
If we think of prices being given by the Black-Scholes formula with $r = 0$, $T = 1$ and $\sigma = 0.25$, then the Black-Scholes price of the digital is $\tilde{D}_1^{BS} = 0.4503$.
Note that over most of the interval $\Omega$, $H'$ is basically a linear function, so that starting with a decent guess for $\tilde{D}_1$ leads almost immediately to the solution.
For example, if we start with the middle digital $\tilde{D}_1^{m} = 0.4503$ (which, curiously, is almost, but not quite, identical to $\tilde{D}_1^{BS}$), the algorithm takes just $2$ steps to find the root to a tolerance of $10^{-9}$.

In Figure \ref{Fig:2b} we have plotted $H(\tilde{D}_1, \tilde{D}_2)$ for
\begin{alignat*}{2}
F &= 100.00,
\\
K_1 &= 80.00, & \quad \tilde{C}_1 &= 22.2656,
\\
K_2 &= 120.00, & \tilde{C}_2 &= 3.7059.
\end{alignat*}
In this set up we have $\Omega = ]0.4640, 0.9717[ \times ]0, 0.4640[$.
Again, starting with middle digital prices $\tilde{D}_1^{m} = 0.7178, \tilde{D}_2^{m} = 0.2320$ leads to the maximum
$H(0.7884, 0.1991) = 4.6208$ in $3$ steps.

\subsection{Some Differences to Buchen and Kelly's Algorithm}

It's easy to suggest a starting point for the root search, for example the middle between the left and right hand sides of
equation \eqref{NoArbitrageDigitals}, or alternatively a single call spread as in equation \eqref{MiddleDigitalPrice},
which amounts to the same thing if the spacing between the strikes is constant.

Note that compared to the Buchen-Kelly algorithm:
\begin{itemize}
\item
We start with a density that already matches all given call prices, and the root-search takes place in the family $\mathcal{G}$ of densities
matching the call prices.
Therefore, at each step in the algorithm, the occurring density has a precise financial interpretation.
\item
The Hessian matrix $H''$ is tridiagonal, so that inverting it is extremely easy and fast.
In \cite{BuchenKelly1996}, the corresponding matrix is given as the covariance matrix between the calls at the different strikes,
and all elements will usually be non-zero.
\item
The dimension of the root-finding problem considered here is one less than in \cite{BuchenKelly1996}.
For example, given the forward $F = \tilde{C}_0$ and a call $\tilde{C}_1$ as in the first scenario above,
this is a one-dimensional problem, not two-dimensional.
\end{itemize}

\section{Comparing Densities for an Increasing Number of Strikes}
\label{ComparingDensities}

When call prices are available at a large number of strikes, then the upper and lower boundaries
given by left and right call spreads, respectively, define rather tight intervals for arbitrage-free digital prices.
Of course, in the limiting case where call prices are available at all strikes, they completely determine the density,
as is well known from \cite{BreedenLitzenberger1978}, and the digital prices are uniquely determined.

When call prices are available at only a small number of strikes, the digital prices can be from rather
large intervals, and densities matching the call prices can differ quite substantially. However, as the number of strikes increases,
these differences become smaller and smaller.
We measure these ``distances'' between the densities using relative entropy.
Note that by the Csisz\'ar-Kullback inequality \cite{ArnoldMarkowichToscaniUnterreiter2000} convergence in terms of relative entropy implies $L^1$-convergence.

For two piecewise exponential densities, the relative entropy can be easily calculated.
If, for $x \in [K_i, K_{i+1}[$, we have $g(x) = \alpha_i e^{\beta_i x}$ and the prior $p(x) = \hat{\alpha}_i e^{\hat{\beta_i} x}$, then
the relative entropy is given by
\begin{equation}
\label{RelativeEntropyForExponentialDensities}
\begin{split}
I(g\|p)
&= \sum_{i=0}^n \ln \left( \frac{\alpha_i}{\hat{\alpha}_i} \right) \frac{\alpha_i}{\beta_i}(e^{\beta_i K_{i+1}} - e^{\beta_i K_i})
\\
&+ \sum_{i=0}^n (\beta_i - \hat{\beta_i})
\left(
\frac{\alpha_i}{\beta_i}(K_{i+1} e^{\beta_i K_{i+1}} - K_i e^{\beta_i K_i}) - \frac{\alpha_i}{\beta_i^2}(e^{\beta_i K_{i+1}} - e^{\beta_i K_i})
\right).
\end{split}
\end{equation}

Of course, from the proof of Theorem \ref{Th:L1Difference}, we know that in case $p$ is the Buchen-Kelly density and $g$ a density
matching the same call prices, then $I(g\|p) = H(p) - H(g)$ holds.

\subsection{A Black-Scholes Example without Volatility Skew}
\label{BlackScholesExample}

In this first example, assume a flat Black-Scholes market given by the data
\begin{equation*}
F = 100.00, \quad r = 0, \quad \sigma = 0.25 \quad \text{and} \quad T = 1,
\end{equation*}
where all call option prices are calculated with the Black-Scholes formula.

First, we calculate the Buchen-Kelly density in four cases, when $3,\ 5,\ 9$ and $17$ strikes are given, respectively,
with spacing decreasing from $40,\ 20,\ 10$ to $5$.

Second, we calculate the entropy maximising density from \cite{NeriSchneider2009} using digital prices obtained via centered call spreads (CCS)
\begin{equation}
\label{MiddleDigitalPrice}
\tilde{D}_i^m := -\frac{\tilde{C}_{i+1} - \tilde{C}_{i-1}}{K_{i+1} - K_{i-1}}, \quad i=2, ..., n-1.
\end{equation}
At left and right endpoints ($i=1,n$), we don't apply this formula and use the Buchen-Kelly digital prices instead.

Third, we calculate the entropy maximising density from \cite{NeriSchneider2009} using Black-Scholes digital prices obtained from the formula
\begin{equation}
\label{BlackScholesDigitalPrice}
\tilde{D}_i^{BS} := N \left( \frac{\ln \left( \frac{F}{K} \right) - \frac{1}{2} \sigma^2 T}{\sigma \sqrt T} \right), \quad i=1, ..., n,
\end{equation}
where $N$ is the standard normal cumulative distribution function.

Digital prices and entropy of the different densities are reported in Table \ref{tab:BS}.

We also calculate the relative entropies with respect to the first density.
It can be seen that, as the number of strikes increases, the approximations given by either \eqref{MiddleDigitalPrice} or \eqref{BlackScholesDigitalPrice} become closer to the Buchen-Kelly digital prices and the relative entropy, obtained using \eqref{RelativeEntropyForExponentialDensities} with the Buchen-Kelly density as the prior density, decreases.

Note also that in this flat market scenario, the digital prices given by the Black-Scholes formula are practically the same as the Buchen-Kelly
prices if five strikes or more are calibrated to.

\begin{sidewaystable}[hbp]
  \scriptsize
  \centering
  \caption{A Flat Black-Scholes World}
    \begin{tabular}{rrrrrrrrrrrrrrrrrrrrrr}
    \addlinespace
    \toprule
          &         &           & K       & 60     & 65     & 70     & 75     & 80     & 85     & 90     & 95     & 100   & 105   & 110   & 115   & 120   & 125   & 130   & 135   & 140   \\
       n  & Entropy & Rel. Ent. & \~C     & 40.145 & 35.346 & 30.719 & 26.336 & 22.266 & 18.565 & 15.272 & 12.401 & 9.948 & 7.889 & 6.190 & 4.811 & 3.706 & 2.832 & 2.149 & 1.620 & 1.214 \\
    \midrule
       3  & 4.616   & 0.000     & BK \~D  &  0.967 &        &        &        &        &        &        &        & 0.465 &       &       &       &       &       &       &       & 0.070 \\
       5  & 4.608   & 0.000     &         &  0.973 &        &        &        &  0.779 &        &        &        & 0.451 &       &       &       & 0.197 &       &       &       & 0.070 \\
       9  & 4.607   & 0.000     &         &  0.974 &        &  0.903 &        &  0.779 &        &  0.617 &        & 0.450 &       & 0.306 &       & 0.196 &       & 0.120 &       & 0.070 \\
      17  & 4.607   & 0.000     &         &  0.973 &  0.945 &  0.904 &  0.847 &  0.779 &  0.700 &  0.617 &  0.532 & 0.450 & 0.374 & 0.306 & 0.247 & 0.196 & 0.154 & 0.120 & 0.093 & 0.070 \\
    \midrule
       3  & 4.613   & 0.003     & CCS \~D &  0.967 &        &        &        &        &        &        &        & 0.487 &       &       &       &       &       &       &       & 0.070 \\
       4  & 4.587   & 0.021     &         &  0.973 &        &        &        &  0.755 &        &        &        & 0.464 &       &       &       & 0.218 &       &       &       & 0.070 \\
       9  & 4.596   & 0.011     &         &  0.974 &        &  0.894 &        &  0.772 &        &  0.616 &        & 0.454 &       & 0.312 &       & 0.202 &       & 0.125 &       & 0.070 \\
      17  & 4.604   & 0.004     &         &  0.973 &  0.943 &  0.901 &  0.845 &  0.777 &  0.699 &  0.616 &  0.532 & 0.451 & 0.376 & 0.308 & 0.248 & 0.198 & 0.156 & 0.121 & 0.093 & 0.070 \\
    \midrule
       3  & 4.614   & 0.002     & BS \~D  &  0.972 &        &        &        &        &        &        &        & 0.450 &       &       &       &       &       &       &       & 0.071 \\
       5  & 4.608   & 0.000     &         &  0.972 &        &        &        &  0.779 &        &        &        & 0.450 &       &       &       & 0.196 &       &       &       & 0.071 \\
       9  & 4.607   & 0.000     &         &  0.972 &        &  0.903 &        &  0.779 &        &  0.617 &        & 0.450 &       & 0.306 &       & 0.196 &       & 0.120 &       & 0.071 \\
      17  & 4.607   & 0.000     &         &  0.972 &  0.945 &  0.903 &  0.847 &  0.779 &  0.700 &  0.617 &  0.532 & 0.450 & 0.374 & 0.306 & 0.247 & 0.196 & 0.154 & 0.120 & 0.093 & 0.071 \\
    \bottomrule
    \end{tabular}
  \label{tab:BS}
\end{sidewaystable}

\subsection{An Example with Call Option Prices on the S\&P 500}
\label{SPXExample}

In the second example, the market is given by CBOE option prices from 10 April 2010 with maturity 18 December 2010 and a forward $F = 1178.00$.

First, we calculate the Buchen-Kelly density in four cases, when $3,\ 5,\ 9$ and $17$ strikes are given, respectively,
with spacing decreasing from USD $200,\ 100,\ 50$ to $25$.

Second, we calculate the entropy maximising density from \cite{NeriSchneider2009} using digital prices obtained from the approximation \eqref{MiddleDigitalPrice}.
Again, at left and right endpoints we use the Buchen-Kelly prices.

Digital prices and entropy of the two densities are reported in Table \ref{tab:SPX}.

We also calculate the relative entropy of the second density with respect to the first density.

Again, it can be seen that, as the number of strikes increases, the approximations given by \eqref{MiddleDigitalPrice} become closer to the
Buchen-Kelly digital prices, and the relative entropy obtained using \eqref{RelativeEntropyForExponentialDensities} with the Buchen-Kelly density
as the prior density decreases.

\begin{sidewaystable}[hbp]
  \tiny
  \centering
  \caption{An Example with Call Option Prices on the S\&P 500}
    \begin{tabular}{rrrrrrrrrrrrrrrrrrrrr}
    \addlinespace
    \toprule
         &         & K         &         & 1000    & 1025    & 1050    & 1075    & 1100    & 1125    & 1150   & 1175   & 1200   & 1225   & 1250   & 1275   & 1300   & 1325   & 1350   & 1375  & 1400  \\
       n & Entropy & Rel. Ent. & \~C     & 207.919 & 186.796 & 166.526 & 146.958 & 128.294 & 110.583 & 93.925 & 78.522 & 64.473 & 51.879 & 40.791 & 31.308 & 23.431 & 17.059 & 12.042 & 8.228 & 5.168 \\
    \midrule
       3 &  6.6363 &    0.0000 & BK \~D  &   0.843 &         &         &         &         &         &        &        &  0.530 &        &        &        &        &        &        &       & 0.095 \\
       5 &  6.6345 &    0.0000 &         &   0.846 &         &         &         &   0.732 &         &        &        &  0.532 &        &        &        &  0.289 &        &        &       & 0.091 \\
       7 &  6.6325 &    0.0000 &         &   0.851 &         &   0.800 &         &   0.729 &         &  0.642 &        &  0.534 &        &  0.411 &        &  0.283 &        &  0.180 &       & 0.095 \\
      17 &  6.6234 &    0.0000 &         &   0.857 &   0.829 &   0.797 &   0.766 &   0.728 &   0.689 &  0.642 &  0.590 &  0.533 &  0.474 &  0.412 &  0.347 &  0.284 &  0.227 &  0.173 & 0.137 & 0.104 \\
    \midrule
       3 &  6.6314 &    0.0049 & CSS \~D &   0.843 &         &         &         &         &         &        &        &  0.507 &        &        &        &        &        &        &       & 0.095 \\
       5 &  6.6266 &    0.0079 &         &   0.846 &         &         &         &  0.717  &         &        &        &  0.524 &        &        &        &  0.297 &        &        &       & 0.091 \\
       7 &  6.6288 &    0.0037 &         &   0.851 &         &   0.796 &         &  0.726  &         &  0.638 &        &  0.531 &        &  0.410 &        &  0.287 &        &  0.183 &       & 0.095 \\
      17 &  6.6217 &    0.0017 &         &   0.857 &   0.828 &   0.797 &   0.765 &  0.728  &   0.687 &  0.641 &  0.589 &  0.533 &  0.474 &  0.411 &  0.347 &  0.285 &  0.228 &  0.177 & 0.137 & 0.104 \\
    \bottomrule
    \end{tabular}
  \label{tab:SPX}
\end{sidewaystable}

Note that in both examples given in \ref{BlackScholesExample} and \ref{SPXExample}, we can observe the pattern
that in-the-money ($K_i < F$) CCS digital price estimates tend to be too low
and out-of-the-money ($K_i > F$) ones too high compared to the digital prices obtained with the Buchen-Kelly density.
One could therefore try to establish better initial guesses, but we will not pursue this here.

\section{Conclusion}
\label{Conclusion}

In this article we study the family $\mathcal{G}$ of piecewise exponential densities matching given call option prices with different strikes
and for a fixed maturity $T$.
Considerations of entropy maximisation lead directly to this particular form, which has the advantages of being easy to work with (distribution functions
can be calculated and inverted analytically, for example) and generating very realistic looking ``volatility smiles''.

We make clear that despite the constraints imposed by the call option prices, the probability assigned by different densities in $\mathcal{G}$ for the underlying asset to lie in a given interval can vary considerably. This will obviously have a significant impact when pricing just about any type of derivative, beginning with digital options.
However, the more such call option prices are observable, the smaller the scope for these variations will become.
So, on the one hand, when only the forward price of the asset and an at-the-money call option price are known, the probability of, for example, the asset finishing in-the-money will change substantially, depending on the density chosen from $\mathcal{G}$.
On the other hand, in the theoretical case where call option prices are available at all strikes, these will be enough to completely determine the density
(the family will consist of just this one member), and therefore also the probability of the asset finishing in-the-money.

Giving digital option prices at the same strikes as for the calls uniquely determines a density in the family $\mathcal{G}$, and we therefore think of it as
being parameterised by these digital prices. An analysis of the entropy function $H$ over this family allows us to show that there is a unique density 	maximising $H$, the {\it Buchen-Kelly density}, which also is the only continuous density in $\mathcal{G}$.

Based on these results we introduce a Newton-Raphson algorithm for finding the Buchen-Kelly density that, in our opinion, is simpler than the one originally proposed,
and we list some of its advantages. Using the gradient of $H$, expressed as a function of digital prices, we give three convergence criteria for the algorithm.

Finally, in two market scenarios, one fictitious and the other with market data from the CBOE,
we illustrate how the densities matching these call prices converge to the Buchen-Kelly density when prices at more and more strikes are given.

\bibliographystyle{plain}
\bibliography{../bibtex/articles,../bibtex/books}

\vspace{2cm}

Cassio Neri and Lorenz Schneider have Ph.D.s in mathematics from Universities Paris IX and VI, respectively.

\end{document}